\documentclass[english,conference]{IEEEtran}
\usepackage[T1]{fontenc}
\usepackage[latin9]{inputenc}
\usepackage{geometry}
\usepackage{amsmath}
\usepackage{amssymb}
\usepackage{graphicx}

\makeatletter
\IEEEoverridecommandlockouts

\usepackage{algorithmic}
\usepackage{amsthm}
\usepackage{subfigure}
\usepackage{float}
\usepackage{amsmath}
\usepackage{graphicx}
\usepackage{setspace}
\usepackage{bbm}

\floatstyle{ruled}
\newfloat{algorithm}{tbp}{loa}
\providecommand{\algorithmname}{Algorithm}
\floatname{algorithm}{\protect\algorithmname}

\newtheorem{proposition}{Proposition}
\newtheorem{corollary}{Corollary}

\newtheorem{lemma}{Lemma}

\usepackage{babel}
\usepackage{upquote}
\usepackage{booktabs}\usepackage{bm}\usepackage{cite}

\newcommand\undermat[2]{%
  \makebox[0pt][l]{$\smash{\underbrace{\phantom{%
    \begin{matrix}#2\end{matrix}}}_{\text{$#1$}}}$}#2}

\makeatother

\usepackage{babel}
\begin{document}

\title{Mobility-Induced Service Migration in Mobile Micro-Clouds}

\author{\IEEEauthorblockN{Shiqiang Wang\IEEEauthorrefmark{1}, Rahul
Urgaonkar\IEEEauthorrefmark{2}, Ting He\IEEEauthorrefmark{2},
Murtaza Zafer\IEEEauthorrefmark{3}\IEEEauthorrefmark{5},
Kevin Chan\IEEEauthorrefmark{4}, and Kin K. Leung\IEEEauthorrefmark{1}}\IEEEauthorblockA{\IEEEauthorrefmark{1}Department
of Electrical and Electronic Engineering, Imperial College London,
United Kingdom} \IEEEauthorblockA{\IEEEauthorrefmark{2}IBM T.
J. Watson Research Center, Yorktown Heights, NY, United States} \IEEEauthorblockA{\IEEEauthorrefmark{3}Samsung
Research America, San Jose, CA, United States} \IEEEauthorblockA{\IEEEauthorrefmark{4}Army
Research Laboratory, Adelphi, MD, United States} \IEEEauthorblockA{Email:
\IEEEauthorrefmark{1}\{shiqiang.wang11, kin.leung\}@imperial.ac.uk,
\IEEEauthorrefmark{2}\{rurgaon, the\}@us.ibm.com,\\ \IEEEauthorrefmark{3}murtaza.zafer.us@ieee.org,
\IEEEauthorrefmark{4}kevin.s.chan.civ@mail.mil}

\thanks{\IEEEauthorrefmark{5} Contributions of the author to this work are not related to his current employment at Samsung Research America.

\copyright  2014 IEEE. Personal use of this material is permitted. Permission from IEEE must be obtained for all other uses, in any current or future media, including reprinting/republishing this material for advertising or promotional purposes, creating new collective works, for resale or redistribution to servers or lists, or reuse of any copyrighted component of this work in other works.}


}

\maketitle
\begin{abstract}
Mobile micro-cloud is an emerging technology in distributed computing,
which is aimed at providing seamless computing/data access to the
edge of the network when a centralized service may suffer from poor
connectivity and long latency. Different from the traditional cloud,
a mobile micro-cloud is smaller and deployed closer to users, typically
attached to a cellular basestation or wireless network access point.
Due to the relatively small coverage area of each basestation or access
point, when a user moves across areas covered by different basestations
or access points which are attached to different micro-clouds, issues
of service performance and service migration become important. In
this paper, we consider such migration issues. We model the general
problem as a Markov decision process (MDP), and show that, in the
special case where the mobile user follows a one-dimensional asymmetric
random walk mobility model, the optimal policy for service migration
is a threshold policy. We obtain the analytical solution for the cost
resulting from arbitrary thresholds, and then propose an algorithm
for finding the optimal thresholds. The proposed algorithm is more
efficient than standard mechanisms for solving MDPs.\end{abstract}

\begin{IEEEkeywords}
Cloud computing, Markov decision process (MDP), mobile micro-cloud,
mobility, service migration, wireless networks
\end{IEEEkeywords}


\section{Introduction}

Cloud technologies have been developing successfully in the past decade,
which enable the centralization of computing and data resources so
that they can be accessed in an on-demand basis by different end users.
Traditionally, clouds are centralized, in the sense that services
are provided by large data-centers that may be located far away from
the user. A user may suffer from poor connectivity and long latency
when it connects to such a centralized service. In recent years, efforts
have been made to distribute the cloud closer to users, to provide
faster access and higher reliability to end users in a particular
geographical area. A notable concept in this regard is the mobile
micro-cloud, where a small cloud consisting of a small set of servers
is attached directly to the wireless communication infrastructure
(e.g., a cellular basestation or wireless access point) to provide
service to users within its coverage. Applications of the mobile micro-cloud
include data and computation offloading for mobile devices \cite{abe2013vtube,ha2013WearableCognitiveAssistance},
which is a complement for the relatively low computational and data
storage capacity of mobile users. It is also beneficial for scenarios
requiring high robustness or high data-processing capability closer
to the user, such as in hostile environments \cite{CloudletHostile}
or for vehicular networks \cite{wang2013VNET}. There are a few other
concepts which are similar to that of the mobile micro-cloud, including
edge computing \cite{Edge-as-a-service}, Cloudlet \cite{CloudletHostile},
and Follow Me Cloud \cite{FollowMeMagazine}. We use the term mobile
micro-cloud throughout this paper.

\begin{figure}
\center{\includegraphics[width=0.8\linewidth]{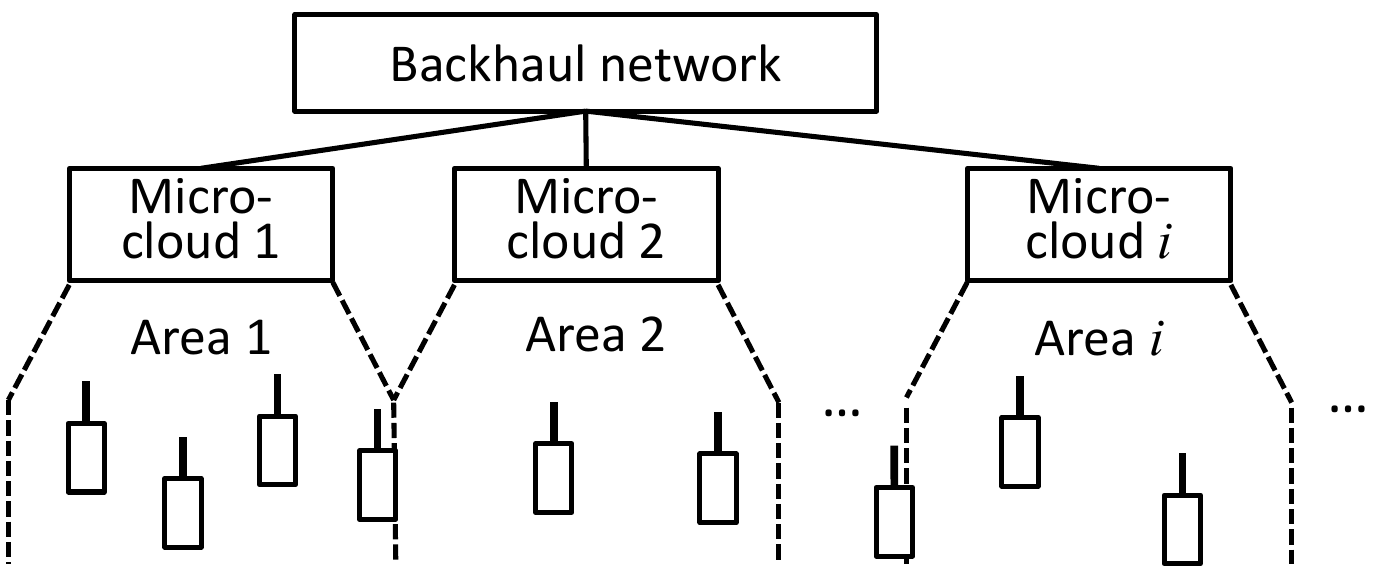}}

\protect\caption{Application scenario with mobile micro-cloud.}
\label{fig:scenario}
\end{figure}

A significant issue in the mobile micro-cloud is service migration
caused by the mobility of users. Because different micro-clouds are
attached to different basestations or access points, a decision needs
to be made on whether and where to migrate the service, when a user
moves outside the service area of a micro-cloud that is providing
its service. Consider the scenario as shown in Fig. \ref{fig:scenario},
which resembles the case where a micro-cloud is connected to a basestation
that covers a particular area, and these micro-clouds are also interconnected
with each other via a backhaul network. When a mobile user moves from
one area to another area, we can either continue to run the service
on the micro-cloud for the previous area, and transmit data to/from
the user via the backhaul network, or we can migrate the service to
the micro-cloud responsible for the new area. In both cases, a cost
is incurred; there is a data transmission cost for the first case,
and a migration cost for the second case.

In the literature, only a few papers have studied the impact of mobility
and its relationship to service migration for mobile micro-clouds.
In \cite{FollowMeGC2013}, analytical results on various performance
factors of the mobile micro-cloud are studied, by assuming a symmetric
2-dimensional (2-D) random walk mobility model. A service migration
procedure based on Markov decision process (MDP) for 1-D random walk
is studied in \cite{MDPFollowMeICC2014}. It mainly focuses on formulating
the problem with MDP, which can then be solved with standard techniques
for solving MDPs. 

In this paper, similarly to \cite{MDPFollowMeICC2014}, we consider
an MDP formulation of the migration problem. In contrast to \cite{MDPFollowMeICC2014},
we propose an optimal threshold policy to solve for the optimal action
of the MDP, which is more efficient than standard solution techniques.
A threshold policy means that we always migrate the service for a
user from one micro-cloud to another when the user is in states bounded
by a particular set of thresholds, and not migrate otherwise. We first
prove the existence of an optimal threshold policy and then propose
an algorithm with polynomial time-complexity for finding the optimal
thresholds. The analysis in this paper can also help us gain new insights
into the migration problem, which set the foundation for more complicated
scenarios for further study in the future. 

The remainder of this paper is organized as follows. In Section II,
we describe the problem formulation. Section III shows that an optimal
threshold policy exists and proposes an algorithm to obtain the optimal
thresholds. Simulation results are shown in Section IV. Section V
draws conclusions.

\section{Problem Formulation}

We consider a 1-D region partitioned into a discrete set of areas,
each of which is served by a micro-cloud, as shown in Fig. \ref{fig:scenario}.
Such a scenario models user mobility on roads, for instance. A time-slotted
system (Fig. \ref{fig:timing}) is considered, which can be viewed
as a sampled version of a continuous-time system, where the sampling
can either be equally spaced over time or occur right after a handoff
instance. 

Mobile users are assumed to follow a 1-D asymmetric random walk mobility
model. In every new timeslot, a node moves with probability $p$ (or
$q$) to the area that is on the right (or left) of its previous area,
it stays in the same area with probability $1-p-q$. If the system
is sampled at handoff instances, then $1-p-q=0$, but we consider
the general case with $0\leq1-p-q\leq1$. Obviously, this mobility
model can be described as a Markov chain. We only focus on a single
mobile user in our analysis; equivalently, we assume that there is
no correlation in the service or mobility among different users. 

The state of the user is defined as the \emph{offset} between the
mobile user location and the location of the micro-cloud running the
service at the beginning of a slot, before possible service migration,
i.e., the state in slot $t$ is defined as $s_{t}=u_{t}-h_{\ensuremath{t}}$,
where $u_{t}$ is the location (index of area) of the mobile user,
and $h_{t}$ the location of the micro-cloud hosting the service.
Note that $s_{t}$ can be zero, positive, or negative. At the beginning
of each timeslot, the current state is observed, and the decision
on whether to migrate the service is made. If migration is necessary,
it happens right after the state observation, i.e., at the beginning
of the timeslot. We assume that the time taken for migration is negligible
compared with the length of a slot. 

\begin{figure}
\center{\includegraphics[width=0.95\linewidth]{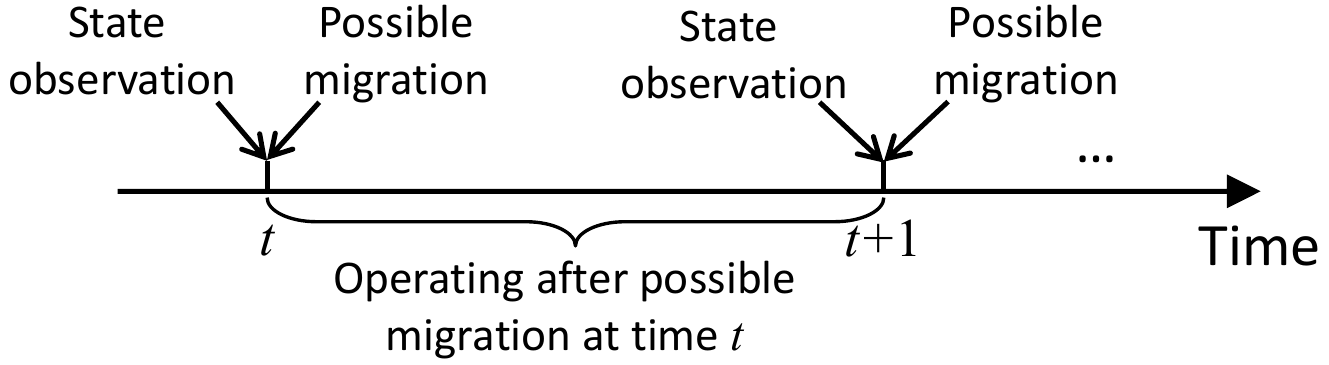}}

\protect\caption{Timing of the proposed mechanism.}
\label{fig:timing}
\end{figure}

We study whether and where to migrate the service when the mobile
user has moved from one area to another. The cost in a single timeslot
$C_{a}(s)$ is defined as the cost under state $s$ when performing
action $a$, where $a$ represents a migration decision for the service
such as no migration or migration to a specified micro-cloud. The
goal is to minimize the discounted sum cost over time. Specifically,
under the current state $s_{0}$, we would like to find a policy $\pi$
that maps each possible state $s$ to an action $a=\pi(s)$ such that
the expected long-term discounted sum cost is minimized, i.e., 
\begin{equation}
V(s_{0})=\min_{\pi}\mathrm{E}\left[\sum_{t=0}^{\infty}\gamma^{t}C_{\pi(s_{t})}\left(s_{t}\right)\Bigg|s_{0}\right]\label{eq:objFunc}
\end{equation}
where $\mathrm{E}\left[\cdot|\cdot\right]$ denotes the conditional
expectation, and $0<\gamma<1$ is a discount factor.

Because we consider a scenario where all micro-clouds are connected
via the backhaul (as shown in Fig. \ref{fig:scenario}), and the backhaul
is regarded as a central entity (which resembles the case for cellular
networks, for example), we consider the following one-timeslot cost
function for taking action $a$ in state $s$ in this paper:
\begin{equation}
C_{a}(s)=\begin{cases}
0, & \textrm{if no migration or data transmission}\\
\beta, & \textrm{if only data transmission}\\
1, & \textrm{if only migration}\\
\beta+1, & \textrm{if both migration and data transmission}
\end{cases}\label{eq:oneTimeslotCostWords}
\end{equation}
Equation (\ref{eq:oneTimeslotCostWords}) is explained as follows.
If the action $a$ under state $s$ causes no migration or data transmission
(e.g., if the node and the micro-cloud hosting the service are in
the same location, i.e., $s=0$, and we do not migrate the service
to another location), we do not need to communicate via the backhaul
network, and the cost is zero. A non-zero cost is incurred when the
node and the micro-cloud hosting the service are in different locations.
In this case, if we do not migrate to the current node location at
the beginning of the timeslot, the data between the micro-cloud and
mobile user need to be transmitted via the backhaul network. This
data transmission incurs a cost of $\beta$. When we perform migration,
we need resources to support migration. The migration cost is assumed
to be $1$, i.e., the cost $C_{a}(s)$ is normalized by the migration
cost. Finally, if both migration and data transmission occur, in which
case we migrate to a location that is different from the current node
location, the total cost is $\beta+1$.

\begin{lemma}\label{lemma:NotToOtherNodes}Migrating to locations
other than the current location of the mobile user is not optimal.\end{lemma}\begin{proof} We
consider an arbitrary trajectory of the mobile user. Denote $t_{u}$
as the first timeslot (starting from the current timeslot) that the
mobile user is in the location indexed $u$. Assume that the user
is currently in location $u_{0}$, then the current timeslot is $t_{u_{0}}$. 

Case 1 -- migrating to location $u\neq u_{0}$ at $t_{u_{0}}$: This
incurs a cost of $\beta+1$ at timeslot $t_{u_{0}}$, because $t_{u}>t_{u_{0}}$
as a node cannot be in two different locations at the same time. Define
a variable $t_{m}\in\left[t_{u_{0}}+1,t_{u}\right]$ being the largest
timeslot index such that we do not perform further migration at timeslots
within the interval $\left[t_{u_{0}}+1,t_{m}-1\right]$, which means
that either we perform migration at $t_{m}$ or we have $t_{m}=t_{u}$.
Then, we have a cost of $\beta$ at each of the timeslots $t\in\left[t_{u_{0}}+1,t_{m}-1\right]$. 

Case 2 -- no migration at $t_{u_{0}}$: In this case, the cost at
each timeslot $t\in\left[t_{u_{0}},t_{m}-1\right]$ is either $\beta$
(if $s_{t}=u_{t}-h_{t}\neq0$) or zero (if $s_{t}=0$). For the timeslot
$t_{m}$, we construct the following policy. If $t_{m}<t_{u}$, we
migrate to the same location as in Case 1, which means that the cost
at $t_{m}$ cannot be larger than that in Case 1. If $t_{m}=t_{u}$,
we migrate to $u$, which can increase the cost at $t_{m}$ by at
most one unit compared with the cost in Case 1. With the above policy,
the costs at timeslots $t>t_{m}$ in Cases 1 and 2 are the same.

The cost at $t_{u_{0}}$ in Case 1 is one unit larger than that in
Case 2, and the cost at $t_{m}$ in Case 1 is at most one unit smaller
than that in Case 2. Because $0<\gamma<1$, Case 2 brings lower discounted
sum cost than Case 1. Therefore, there exists a better policy than
migrating to $u\neq u_{0}$ at $t_{u_{0}}$. This holds for any movement
pattern of the mobile user, and it ensures that the cost in any timeslot
is either $0$, $1$, or $\beta$. \end{proof}

\begin{figure*}
\center{\includegraphics[width=0.9\linewidth]{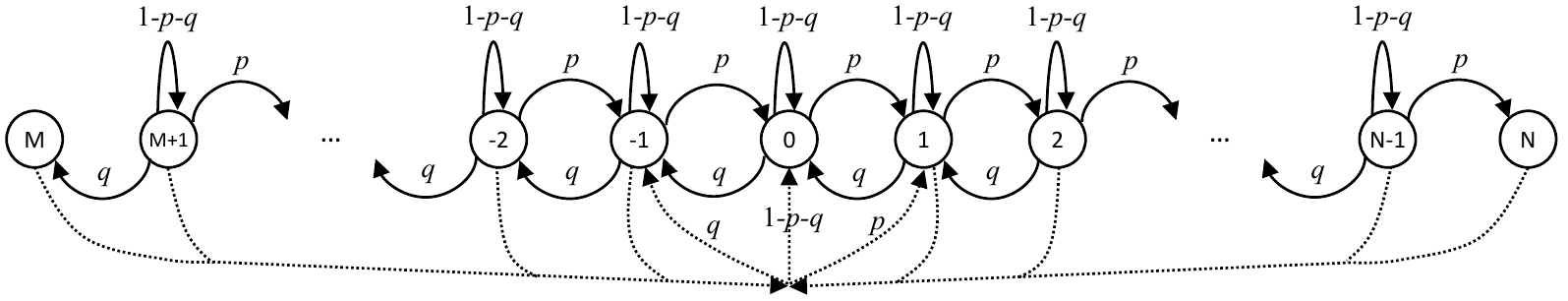}}

\protect\caption{MDP model for service migration. The solid lines denote transition
under action $a=0$ and the dotted lines denote transition under action
$a=1$. When taking action $a=1$ from any state, the next state is
$s=-1$ with probability $q$, $s=0$ with probability $1-p-q$, or
$s=1$ with probability $p$. }
\label{fig:states} 
\end{figure*}

From Lemma \ref{lemma:NotToOtherNodes}, we only have two candidate
actions, which are migrating to the current user location or not migrating.
This simplifies the action space to two actions: a migration action,
denoted as $a=1$; and a no-migration action, denoted as $a=0$. In
practice, there is usually a limit on the maximum allowable distance
between the mobile user and the micro-cloud hosting its service for
the service to remain usable. We model this limitation by a maximum
negative offset $M$ and a maximum positive offset $N$ (where $M<0,N>0$),
such that the service must be migrated ($a\equiv1$) when the node
enters state $M$ or $N$. This implies that, although the node can
move in an unbounded space, the state space of our MDP for service
control is finite. The overall transition diagram of the resulting
MDP is illustrated in Fig. \ref{fig:states}. Note that because each
state transition is the concatenated effect of (possible) migration
and node movement, and the states are defined as the offset between
node and host location, the next state after taking a migration action
is either $-1$, $0$, or $1$.

With the above considerations, the cost function in (\ref{eq:oneTimeslotCostWords})
can be modified to the following:
\begin{equation}
C_{a}(s)=\begin{cases}
0, & \textrm{if }s=0\\
\beta, & \textrm{if }s\neq0,M<s<N,a=0\\
1, & \textrm{if }s\neq0,M\leq s\leq N,a=1
\end{cases}\label{eq:oneTimeslotCostActions}
\end{equation}

With the one-timeslot cost defined as in (\ref{eq:oneTimeslotCostActions}),
we obtain the following Bellman's equations for the discounted sum
cost when respectively taking action $a=0$ and $a=1$:
\begin{equation}
V(s|a=0)=\begin{cases}
\gamma\sum_{j=-1}^{1}p_{0j}V(j), & \textrm{if }s=0\\
\beta+\gamma\sum_{j=s-1}^{s+1}p_{sj}V(j), & \textrm{if }s\neq0,M\!\!<\!\! s\!\!<\!\! N
\end{cases}\label{eq:balanceEqu0}
\end{equation}
\begin{equation}
V(s|a=1)=\begin{cases}
\gamma\sum_{j=-1}^{1}p_{0j}V(j), & \textrm{if }s=0\\
1+\gamma\sum_{j=-1}^{1}p_{0j}V(j), & \textrm{if }s\neq0,M\!\!\leq\!\! s\!\!\leq\!\! N
\end{cases}\label{eq:balanceEqu1}
\end{equation}
where $p_{ij}$ denotes the (one-step) transition probability from
state $i$ to state $j$, their specific values are related to parameters
$p$ and $q$ as defined earlier. The optimal cost $V(s)$ is 
\begin{equation}
V(s)=\begin{cases}
\min\{V(s|a=0),V(s|a=1)\}, & \textrm{if }M\!<\! s\!<\! N\\
V(s|a=1), & \textrm{if }s\!=\! M\textrm{ or }s\!=\! N
\end{cases}\label{eq:optCost}
\end{equation}

\section{Optimal Threshold Policy}

\subsection{Existence of Optimal Threshold Policy}

We first show that there exists a threshold policy which is optimal
for the MDP in Fig. \ref{fig:states}.

\begin{proposition}\label{prop:threshExistence}There exists a threshold
policy $(k_{1},k_{2})$, where $M<k_{1}\leq0$ and $0\leq k_{2}<N$,
such that when $k_{1}\leq s\leq k_{2}$, the optimal action for state
$s$ is $a^{*}(s)=0$, and when $s<k_{1}$ or $s>k_{2}$, $a^{*}(s)=1$.\end{proposition}\begin{proof} It
is obvious that different actions for state zero $a(0)=0$ and $a(0)=1$
are essentially the same, because the mobile user and the micro-cloud
hosting its service are in the same location under state zero, either
action does not incur cost and we always have $C_{a(0)}(0)=0$. Therefore,
we can conveniently choose $a^{*}(0)=0$. 

In the following, we show that, if it is optimal to migrate at $s=k_{1}-1$
and $s=k_{2}+1$, then it is optimal to migrate at all states $s$
with $M\leq s\leq k_{1}-1$ or $k_{2}+1\leq s\leq N$. We relax the
restriction that we always migrate at states $M$ and $N$ for now,
and later discuss that the results also hold for the unrelaxed case.
We only focus on $k_{2}+1\leq s\leq N$, because the case for $M\leq s\leq k_{1}-1$
is similar.

If it is optimal to migrate at $s=k_{2}+1$, we have 
\begin{equation}
V(k_{2}+1|a=1)\leq\beta\sum_{t=0}^{\infty}\gamma^{t}=\frac{\beta}{1-\gamma}\label{eq:threshProofIneq}
\end{equation}
where the right hand-side of (\ref{eq:threshProofIneq}) is the discounted
sum cost of a never-migrate policy supposing that the user never returns
back to state zero when starting from state $s=k_{2}+1$. This cost
is an upper bound of the cost incurred from any possible state-transition
path without migration, and migration cannot bring higher cost than
this because otherwise it contradicts with the presumption that it
is optimal to migrate at $s=k_{2}+1$.

Suppose we do not migrate at state $s'$ where $k_{2}+1<s'\leq N$,
then we have a (one-timeslot) cost of $\beta$ in each timeslot until
the user reaches a migration state (i.e., a state at which we perform
migration). From (\ref{eq:balanceEqu1}), we know that $V(s|a=1)$
is constant for $s\neq0$. Therefore, any state-transition path $L$
starting from state $s'$ has a discounted sum cost of
\[
V_{L}(s')=\beta\sum_{t=0}^{t_{m}-1}\gamma^{t}+\gamma^{t_{m}}V(k_{2}+1|a=1)
\]
where $t_{m}>0$ is a parameter representing the first timeslot that
the user is in a migration state after reaching state $s'$ (assuming
that we reach state $s'$ at $t=0$), which is dependent on the state-transition
path $L$. We have
\begin{align*}
 & V_{L}(s')-V(k_{2}+1|a=1)\\
 & =\beta\frac{\left(1-\gamma^{t_{m}}\right)}{1-\gamma}-\left(1-\gamma^{t_{m}}\right)V(k_{2}+1|a=1)\\
 & =\left(1-\gamma^{t_{m}}\right)\left(\frac{\beta}{1-\gamma}-V(k_{2}+1|a=1)\right)\geq0
\end{align*}
where the last inequality follows from (\ref{eq:threshProofIneq}).
It follows that, for any possible state-transition path $L$, $V_{L}(s')\geq V(k_{2}+1|a=1)$.
Hence, it is always optimal to migrate at state $s'$, which brings
cost $V(s'|a=1)=V(k_{2}+1|a=1)$.

The result also holds with the restriction that we always migrate
at states $M$ and $N$, because no matter what thresholds $(k_{1},k_{2})$
we have for the relaxed problem, migrating at states $M$ and $N$
always yield a threshold policy. \end{proof}

Proposition \ref{prop:threshExistence} shows the existence of an
optimal threshold policy. The optimal threshold policy exists for
arbitrary values of $M$, $N$, $p$, and $q$.

\subsection{Simplifying the Cost Calculation }

The existence of the optimal threshold policy allows us simplify the
cost calculation, which helps us develop an algorithm that has lower
complexity than standard MDP solution algorithms. When the thresholds
are given as $\left(k_{1},k_{2}\right)$, the value updating function
(\ref{eq:optCost}) is changed to the following: 
\begin{equation}
V(s)=\begin{cases}
V(s|a=0), & \textrm{if }k_{1}\leq s\leq k_{2}\\
V(s|a=1), & \textrm{otherwise}
\end{cases}\label{eq:optCostGivenTh}
\end{equation}
From (\ref{eq:balanceEqu0}) and (\ref{eq:balanceEqu1}), we know
that, for a given policy with thresholds $(k_{1},k_{2})$, we only
need to compute $V(s)$ with $k_{1}-1\leq s\leq k_{2}+1$, because
the values of $V(s)$ with $s\leq k_{1}-1$ are identical, and the
values of $V(s)$ with $s\geq k_{2}-1$ are also identical. Note that
we always have $k_{1}-1\geq M$ and $k_{2}+1\leq N$, because $k_{1}>M$
and $k_{2}<N$ as we always migrate when at states $M$ and $N$. 

Define
\begin{equation}
\mathbf{v}_{\left(k_{1},k_{2}\right)}=\left[V(k_{1}-1)\,\, V(k_{1})\cdots V(0)\cdots V(k_{2})\,\, V(k_{2}+1)\right]^{\mathrm{T}}\label{eq:defMatrixVk}
\end{equation}
\begin{align} 
\mathbf{c}_{\left(k_{1},k_{2}\right)}=\begin{array}{c} 
\left[\begin{array}{ccccccccc} 
1 & \undermat{-k_1\textrm{ elements}}{\beta & \!\!\cdots\!\! & \beta} & 0 & \undermat{k_2\textrm{ elements}}{\beta & \!\!\cdots\!\! & \beta} & 1
\end{array}\right]^{\mathrm{T}}
\end{array}
\nonumber  \\ \,
\label{eq:defMatrixC} 
\end{align}
\begin{equation}
\mathbf{P}_{\left(k_{1},k_{2}\right)}'=\left[\!\!\begin{array}{ccccc}
p_{0,k_{1}-1} & \cdots & p_{00} & \cdots & p_{0,k_{2}+1}\\
p_{k_{1},k_{1}-1} & \cdots & p_{k_{1},0} & \cdots & p_{k_{1},k_{2}+1}\\
\vdots &  & \vdots &  & \vdots\\
p_{0,k_{1}-1} & \cdots & p_{00} & \cdots & p_{0,k_{2}+1}\\
\vdots &  & \vdots &  & \vdots\\
p_{k_{2},k_{1}-1} & \cdots & p_{k_{2},0} & \cdots & p_{k_{2},k_{2}+1}\\
p_{0,k_{1}-1} & \cdots & p_{00} & \cdots & p_{0,k_{2}+1}
\end{array}\!\!\right]\label{eq:defMatrixP'}
\end{equation}
where superscript $\mathrm{T}$ denotes the transpose of the matrix. 

Then, (\ref{eq:balanceEqu0}) and (\ref{eq:balanceEqu1}) can be rewritten
as
\begin{equation}
\mathbf{v}_{\left(k_{1},k_{2}\right)}=\mathbf{c}_{\left(k_{1},k_{2}\right)}+\gamma\mathbf{P}_{\left(k_{1},k_{2}\right)}'\mathbf{v}_{\left(k_{1},k_{2}\right)}\label{eq:vectorBalanceEqu}
\end{equation}
The value vector $\mathbf{v}_{\left(k_{1},k_{2}\right)}$ can be obtained
by
\begin{equation}
\mathbf{v}_{\left(k_{1},k_{2}\right)}=\left(\mathbf{I}-\gamma\mathbf{P}_{\left(k_{1},k_{2}\right)}'\right)^{-1}\mathbf{c}_{\left(k_{1},k_{2}\right)}\label{eq:evalVk}
\end{equation}

The matrix $\left(\mathbf{I}-\gamma\mathbf{P}_{\left(k_{1},k_{2}\right)}'\right)$
is invertible for $0<\gamma<1$, because in this case there exists
a unique solution for $\mathbf{v}_{\left(k_{1},k_{2}\right)}$ from
(\ref{eq:vectorBalanceEqu}). Equation (\ref{eq:evalVk}) can be computed
using Gaussian elimination that has a complexity of $O\left(\left(|M|+N\right)^{3}\right)$.
However, noticing that $\mathbf{P}_{\left(k_{1},k_{2}\right)}'$ is
a sparse matrix (because $p_{ij}=0$ for $|j-i|>1$), there can exist
more efficient algorithms to compute (\ref{eq:evalVk}).

\subsection{Algorithm for Finding the Optimal Thresholds }

To find the optimal thresholds, we can perform a search on values
of $\left(k_{1},k_{2}\right)$. Further, because an increase/decrease
in $V(s)$ for some $s$ increases/decreases each element in the cost
vector $\mathbf{v}$ due to cost propagation following balance equations
(\ref{eq:balanceEqu0}) and (\ref{eq:balanceEqu1}), we only need
to minimize $V(s)$ for a specific state $s$. We propose an algorithm
for finding the optimal thresholds, as shown in Algorithm \ref{alg:thresh},
which is a modified version of the standard policy iteration mechanism
\cite[Ch. 3]{powell2007approximate}.

\begin{algorithm} 
\caption{Modified policy iteration algorithm for finding the optimal thresholds} 
\label{alg:thresh} 
\begin{algorithmic}[1] 
\STATE Initialize $k^*_1 \leftarrow 0$, $k^*_2 \leftarrow 0$
\REPEAT \label{AlgLargeLoopStart}
\STATE $k'^*_1\leftarrow k^*_1$, $k'^*_2\leftarrow k^*_2$ //record previous thresholds

\STATE Construct  $\mathbf{c}_{\left(k^*_{1},k^*_{2}\right)}$  and $\mathbf{P}'_{\left(k^*_{1},k^*_{2}\right)}$ according to (\ref{eq:defMatrixC}) and (\ref{eq:defMatrixP'}) \label{AlgExactSolutionStart}
\STATE Evaluate $\mathbf{v}_{\left(k^*_{1},k^*_{2}\right)}$ from (\ref{eq:evalVk})
\STATE Extend $\mathbf{v}_{\left(k^*_{1},k^*_{2}\right)}$ to obtain $V(s)$ for all $M \leq s \leq N$  \label{AlgExactSolutionEnd}

\FOR {$i=1,2$}
\IF {$i=1$} \label{AlgSearchDirectionStart}
\IF {$1+\gamma\sum_{j=-1}^{1}p_{0 j}V(j) < V(k^*_1)$} \label{AlgCheckk1}
\STATE dir $\leftarrow 1$, loopVec $\leftarrow[k^*_1+1,k^*_1+2,...,0]$
\STATE $k^*_1 \leftarrow k^*_1+1$
\ELSE
\STATE dir $\leftarrow 0$, loopVec $\leftarrow[k^*_1-1,k^*_1-2,...,M+1]$
\ENDIF
\ELSIF {$i=2$}
\IF {$1+\gamma\sum_{j=-1}^{1}p_{0 j}V(j) < V(k^*_2)$} \label{AlgCheckk2}
\STATE dir $\leftarrow 1$, loopVec $\leftarrow[k^*_2-1,k^*_2-2,...,0]$
\STATE $k^*_2 \leftarrow k^*_2-1$
\ELSE
\STATE dir $\leftarrow 0$, loopVec $\leftarrow[k^*_2+1,k^*_2+2,...,N-1]$
\ENDIF
\ENDIF \label{AlgSearchDirectionEnd}

\FOR {$k_i=$ each value in loopVec} \label{AlgChangeThreshStart}
\IF {dir $=0$} 
\IF {$\beta+\gamma\sum_{j=k_i-1}^{k_i +1}p_{k_i,j}V(j) < V(k_i)$} \label{AlgChangeCheckRev0}
\STATE $k^*_i\leftarrow k_i$
\ELSIF {$\beta+\gamma\sum_{j=k_i-1}^{k_i +1}\!p_{k_i,j}V(j)\! >\! V(k_i)$} \label{AlgChangeCheckRev0_larger}
\STATE \textbf{exit for}
\ENDIF
\ELSIF {dir $=1$}
\IF {$1+\gamma\sum_{j=-1}^{1}p_{0 j}V(j) < V(k_i)$} \label{AlgChangeCheckRev1}
\STATE $k^*_i\leftarrow k_i-\textrm{sign}(k_i)$
\ELSIF {$1+\gamma\sum_{j=-1}^{1}p_{0 j}V(j) > V(k_i)$}  \label{AlgChangeCheckRev1_larger}
\STATE \textbf{exit for}
\ENDIF
\ENDIF 
\ENDFOR  \label{AlgChangeThreshEnd}

\ENDFOR 

\UNTIL{$k^*_1= k'^*_1$ and $k^*_2= k'^*_2$}
\RETURN $k^*_1$, $k^*_2$  
\end{algorithmic} 
\end{algorithm} 

Algorithm \ref{alg:thresh} is explained as follows. We keep iterating
until the thresholds no longer change, which implies that the optimal
thresholds have been found. The thresholds $\left(k_{1}^{*},k_{2}^{*}\right)$
are those obtained from each iteration. 

Lines \ref{AlgExactSolutionStart}--\ref{AlgExactSolutionEnd} compute
$V(s)$ for all $s$ under the given thresholds $\left(k_{1}^{*},k_{2}^{*}\right)$.
Then, Lines \ref{AlgSearchDirectionStart}--\ref{AlgSearchDirectionEnd}
determine the search direction for $k_{1}$ and $k_{2}$. Because
$V(s)$ in each iteration is computed using the current thresholds
$\left(k_{1}^{*},k_{2}^{*}\right)$, we have actions $a(k_{1}^{*})=a(k_{2}^{*})=0$,
and (\ref{eq:balanceEqu0}) is automatically satisfied when replacing
its left hand-side with $V(k_{1}^{*})$ or $V(k_{2}^{*})$. Lines
\ref{AlgCheckk1} and \ref{AlgCheckk2} check whether iterating according
to (\ref{eq:balanceEqu1}) can yield lower cost. If it does, it means
that migrating is a better action at state $k_{1}^{*}$ (or $k_{2}^{*}$),
which also implies that we should migrate at states $s$ with $M\leq s\leq k_{1}^{*}$
(or $k_{2}^{*}\leq s\leq N$) according to Proposition \ref{prop:threshExistence}.
In this case, $k_{1}^{*}$ (or $k_{2}^{*}$) should be set closer
to zero, and we search through those thresholds that are closer to
zero than $k_{1}^{*}$ (or $k_{2}^{*}$). If Line \ref{AlgCheckk1}
(or Line \ref{AlgCheckk2}) is not satisfied, according to Proposition
\ref{prop:threshExistence}, it is good not to migrate at states $s$
with $k_{1}^{*}\leq s\leq0$ (or $0\leq s\leq k_{2}^{*}$), so we
search $k_{1}$ (or $k_{2}$) to the direction approaching $M$ (or
$N$), to see whether it is good not to migrate under those states. 

Lines \ref{AlgChangeThreshStart}--\ref{AlgChangeThreshEnd} adjust
the thresholds. If we are searching toward state $M$ or $N$ and
Line \ref{AlgChangeCheckRev0} is satisfied, it means that it is better
not to migrate under this state ($k_{i}$), and we update the threshold
to $k_{i}$. When Line \ref{AlgChangeCheckRev0_larger} is satisfied,
it means that it is better to migrate at state $k_{i}$. According
to Proposition \ref{prop:threshExistence}, we should also migrate
at any state closer to $M$ or $N$ than state $k_{i}$, thus we exit
the loop. If we are searching toward state zero and Line \ref{AlgChangeCheckRev1}
is satisfied, it is good to migrate under this state ($k_{i}$), therefore
the threshold is set to one state closer to zero ($k_{i}-\textrm{sign}(k_{i})$).
When Line \ref{AlgChangeCheckRev1_larger} is satisfied, we should
not migrate at state $k_{i}$. According to Proposition \ref{prop:threshExistence},
we should also not migrate at any state closer to zero than state
$k_{i}$, and we exit the loop.

\begin{proposition} \label{prop:SingleChangePolicyIteration} The
threshold-pair $\left(k_{1}^{*},k_{2}^{*}\right)$ is different in
every iteration of the loop starting at Line \ref{AlgLargeLoopStart},
otherwise the loop terminates. \end{proposition}\begin{proof}  The
loop starting at Line \ref{AlgLargeLoopStart} changes $k_{1}^{*}$
and $k_{2}^{*}$ in every iteration so that $V(s)$ for all $s$ become
smaller. It is therefore impossible that $k_{1}^{*}$ and $k_{2}^{*}$
are the same as in one of the previous iterations and at the same
time reduce the value of $V(s)$, because $V(s)$ computed from (\ref{eq:evalVk})
is the stationary cost value for thresholds $\left(k_{1}^{*},k_{2}^{*}\right)$.
The only case when $\left(k_{1}^{*},k_{2}^{*}\right)$ are the same
as in the previous iteration (which does not change $V(s)$) terminates
the loop. \end{proof}

\begin{corollary} \label{cor:AlgComplexity} The number of iterations
in Algorithm \ref{alg:thresh} is $O(|M|N)$. \end{corollary}\begin{proof} 
According to Proposition \ref{prop:SingleChangePolicyIteration},
there can be at most $|M|N+1$ iterations in the loop starting at
Line \ref{AlgLargeLoopStart}. \end{proof}

If we use Gaussian elimination to compute (\ref{eq:evalVk}), the
time-complexity of Algorithm~\ref{alg:thresh} is $O\left(|M|N\left(|M|+N\right)^{3}\right)$. 

\begin{figure*}
\center{\subfigure[]{\includegraphics[width=0.28\linewidth]{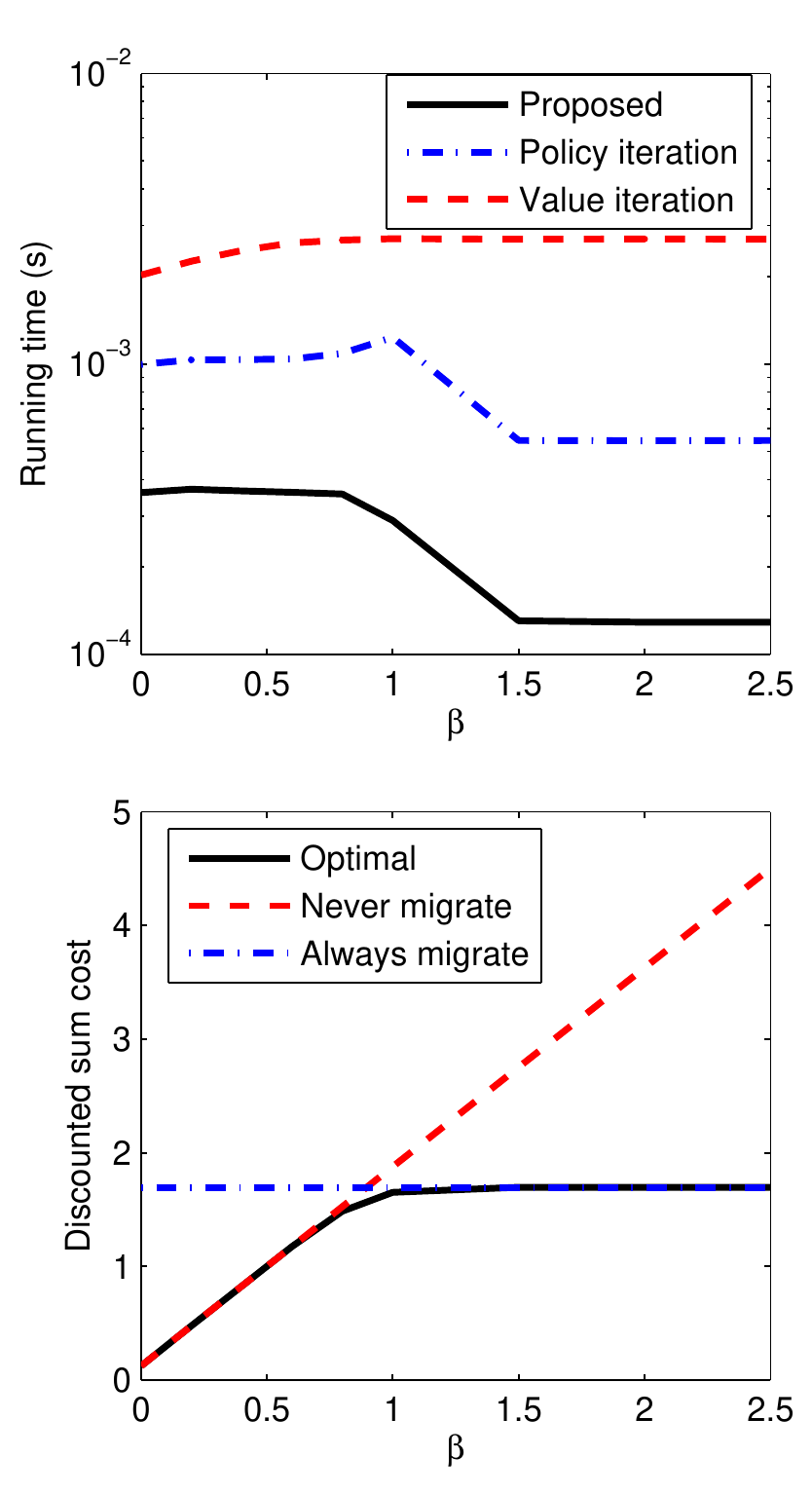}} \subfigure[]{\includegraphics[width=0.28\linewidth]{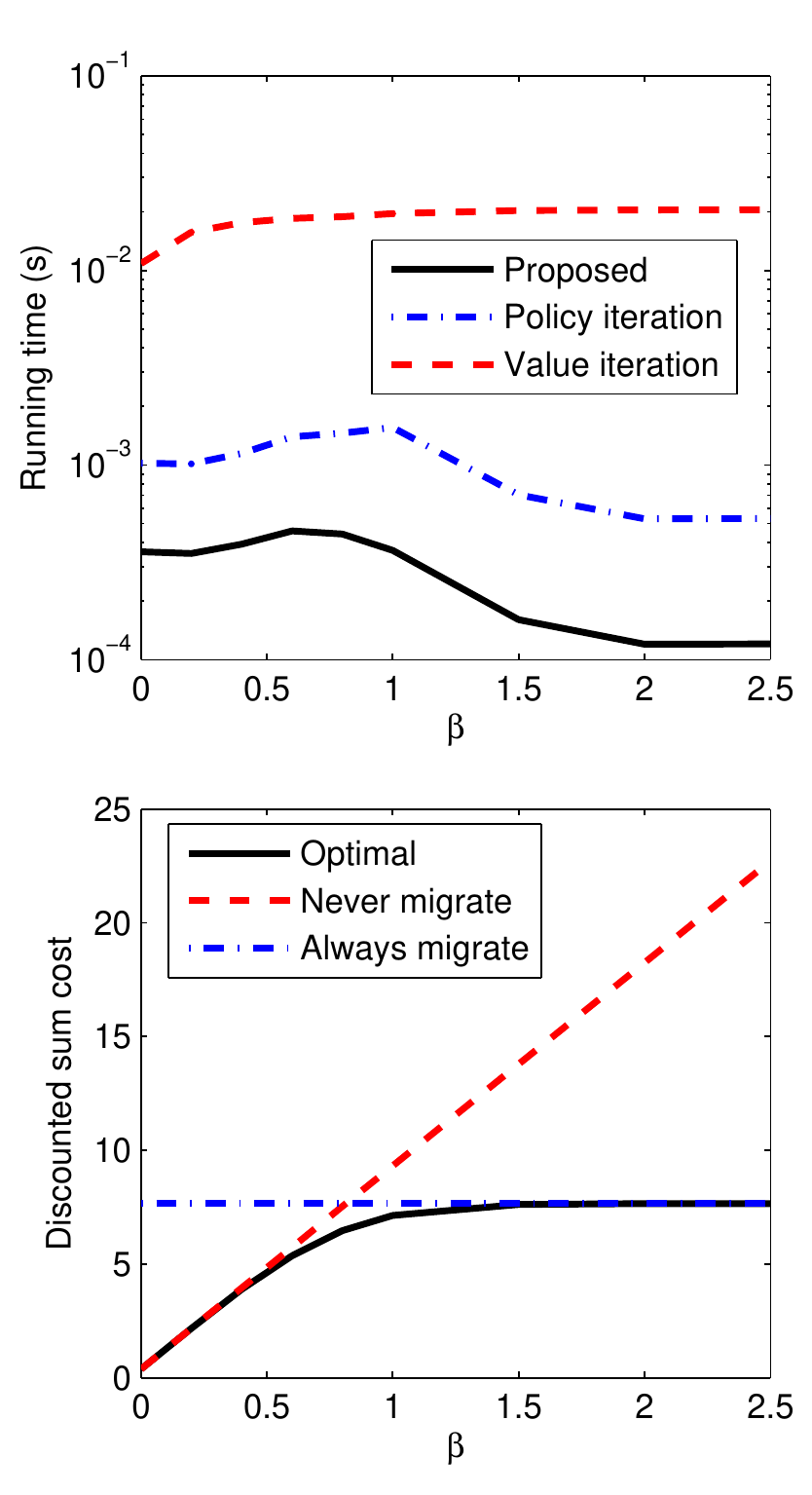}} \subfigure[]{\includegraphics[width=0.28\linewidth]{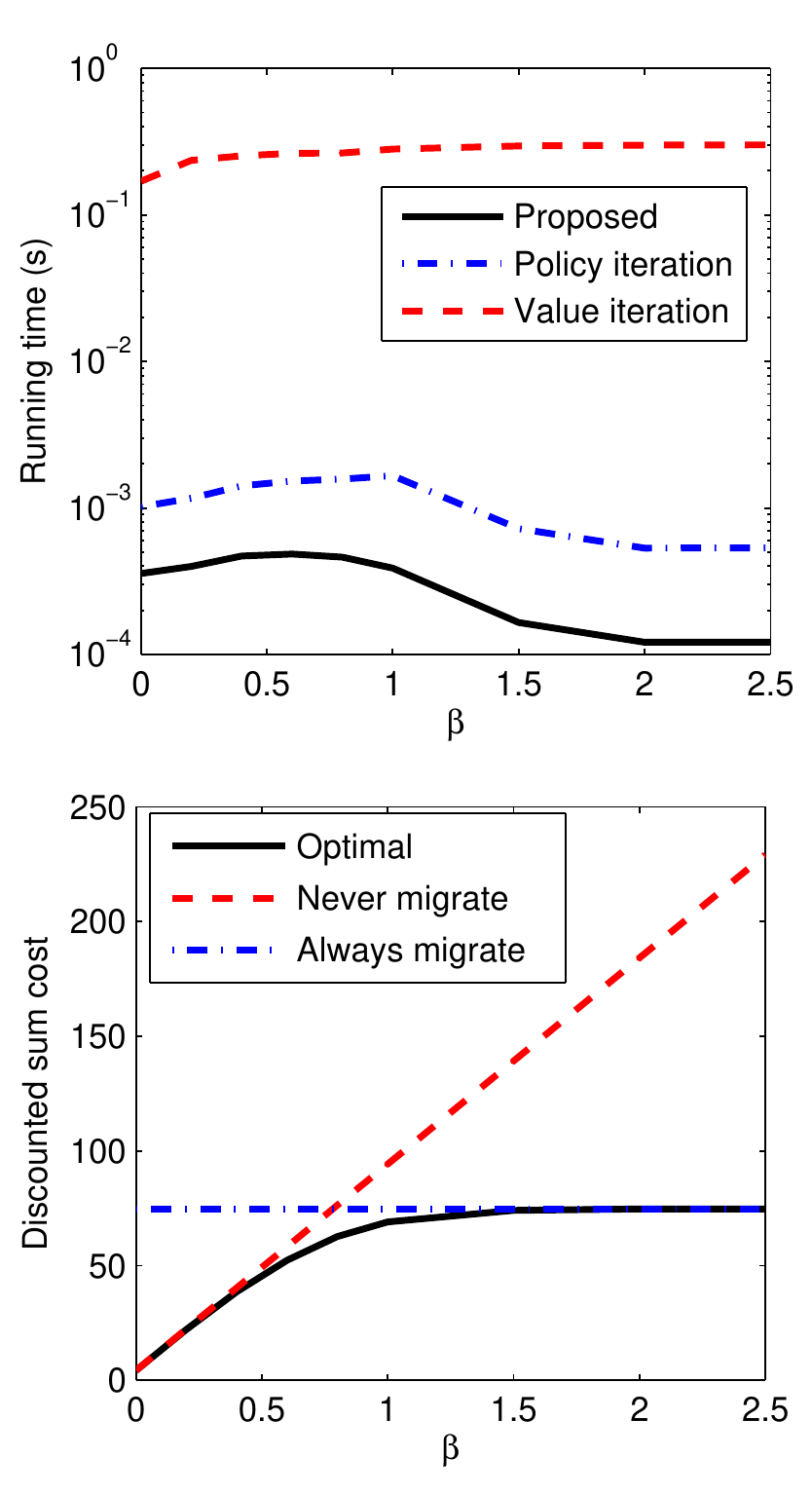}}}

\protect\caption{Performance under different $\beta$: (a) $\gamma=0.5$, (b) $\gamma=0.9$,
(c) $\gamma=0.99$.}
\label{fig:simBeta} \vspace{-0.2in}
\end{figure*}

\section{Simulation Results}

We compare the proposed threshold method with the standard value iteration
and policy iteration methods \cite[Ch. 3]{powell2007approximate}.
Simulations are run on MATLAB, on a computer with 64-bit Windows 7,
Intel Core i7-2600 CPU, and 8GB memory. The value iteration terminates
according to an error bound of $\epsilon=0.1$ in the discounted sum
cost. Note that the proposed method and the standard policy iteration
method always produce the optimal cost. The number of states $|M|=N=10$.
The transition probabilities $p$ and $q$ are randomly generated.
Simulations are run with $1000$ different random seeds in each setting
to obtain the average performance. The running time and the discounted
sum costs under different values of $\beta$ are shown in Fig. \ref{fig:simBeta}.

The results show that the proposed method always has lowest running
time, and the running time of the standard policy iteration method
is $2$ to $5$ times larger than that of the proposed algorithm,
while the value iteration approach consumes longer time. This is because
the proposed algorithm simplifies the solution search procedure compared
with standard mechanisms. The results also show that the proposed
method can provide the optimal cost compared with a never-migrate
(except for states $M$ and $N$) or always-migrate policy. It is
also interesting to observe that the optimal cost approaches the cost
of a never-migrate policy when $\beta$ is small, and it approaches
the cost of an always-migrate policy when $\beta$ is large. Such
a result is intuitive, because a small $\beta$ implies a small data
transmission cost, and when $\beta$ is small enough, then it is not
really necessary to migrate; when $\beta$ is large, the data transmission
cost is so large so that it is always good to migrate to avoid data
communication via the backhaul network.

\section{Conclusion and Future Work}

In this paper, we have proposed a threshold policy-based mechanism
for service migration in mobile micro-clouds. We have shown the existence
of optimal threshold policy and proposed an algorithm for finding
the optimal thresholds. The proposed algorithm has polynomial time-complexity
which is independent of the discount factor $\gamma$. This is promising
because the time-complexity of standard algorithms for solving MDPs,
such as value iteration or policy iteration, are generally dependent
on the discount factor, and they can only be shown to have polynomial
time-complexity when the discount factor is regarded as a constant%
\footnote{This is unless the MDP is deterministic, which is not the case in
this paper.%
} \cite{post2013simplex}. Although the analysis in this paper is based
on 1-D random walk of mobile users, it can serve as a theoretical
basis for more complicated scenarios, such as 2-D user-mobility, in
the future.

\section*{Acknowledgment}

This research was partly sponsored by the U.S. Army Research Laboratory
and the U.K. Ministry of Defence and was accomplished under Agreement
Number W911NF-06-3-0001. The views and conclusions contained in this
document are those of the author(s) and should not be interpreted
as representing the official policies, either expressed or implied,
of the U.S. Army Research Laboratory, the U.S. Government, the U.K.
Ministry of Defence or the U.K. Government. The U.S. and U.K. Governments
are authorized to reproduce and distribute reprints for Government
purposes notwithstanding any copyright notation hereon.

\bibliographystyle{IEEEtran}
\bibliography{IEEEabrv,MappingProblemWMobility}

\end{document}